\documentclass[11pt]{article}
\usepackage{times}
\usepackage{amsfonts}
\usepackage{amsmath}
\usepackage{amssymb}
\usepackage{amsthm}
\usepackage{latexsym}

\usepackage{fullpage}

\usepackage{latexsym}
\usepackage{amsmath}
\usepackage{amstext}
\usepackage{amsfonts}
\usepackage{amsopn}
\usepackage{float}
\usepackage{url}
\usepackage{color}

\newtheorem{theorem}{Theorem}
\newtheorem{lemma}[theorem]{Lemma}

\newtheorem{problem}[theorem]{Problem}

\newtheorem{definition}[theorem]{Definition}
\floatstyle{boxed}

\newenvironment{ProofDummyEnv}{}{}

\newenvironment{none}{}{}

{\QED\end{ProofDummyEnv}}

\newcommand{\hide}[1]{}

\newcommand{\QED}{\nopagebreak\hfill $\Box$}

\newcommand{\Ex}[1]{E\!\left[#1\right]}

\providecommand{\abs}[1]{\lvert#1\rvert}

\providecommand{\norm}[1]{\left\lVert#1\right\rVert}
\providecommand{\snorm}[1]{\bigl\lVert#1\bigr\rVert}
\newcommand{\mnorm}[1]{%
  \vert\kern-0.9pt\vert\kern-0.9pt\vert#1
  \vert\kern-0.9pt\vert\kern-0.9pt\vert
}

\newcommand{\vareps}{\varepsilon}
 \DeclareMathOperator{\spn}{span}

{\unskip\nobreak\hskip 2em plus 1fil\nobreak\QED\parfillskip=0pt%
\endtrivlist}

\def\enumproof{\rm \trivlist \item[\hskip \labelsep{\bf Proof. }]
\begin{enumerate}}
\def\endenumproof{\QED \end{enumerate} \endtrivlist}

\def\descproof{\rm \trivlist \item[\hskip \labelsep{\bf Proof. }]
\begin{description}}
\def\enddescproof{\QED \end{description} \endtrivlist}

\renewcommand{\Pr}[1]{
{\rm Pr}\left[{#1}\right] }

\newcommand{\row}[2]{\ensuremath{{#1}_{#2 \star}}}
\newcommand{\col}[2]{\ensuremath{{#1}_{\star #2}}}
\newcommand{\Sa}{\ensuremath{S}}
\newcommand{\Sb}{\ensuremath{T}}

\long\def\symbolfootnote[#1]#2{\begingroup%
\def\thefootnote{\fnsymbol{footnote}}\footnote[#1]{#2}\endgroup}

\definecolor{adg}{rgb}{0.5,0.60,0.9}
\definecolor{frankcolor}{rgb}{0.15,0.7,0.35}
\newcommand{\outline}[1]{}

\newcommand{\cnv}{\boldsymbol{\mu}}

\newcommand{\eps}{\epsilon}
\def\vareps{\varepsilon}

\newcommand{\Real}{\mathbb{R}}

\newcommand{\A}{A}
\newcommand{\U}{U}

\usepackage{epsfig}
\usepackage{amsmath}
\usepackage{amsbsy}

\newcommand{\omt}[1]{}

\newcommand{\cn}[1]{\cnv_{1}}

\def\OPT{\text{\sc opt}}

\newcommand{\para}[1]{\smallskip \noindent {\bf #1.}\ }

\begin{document}

\title{\Large \bf Sampling Algorithms and Coresets for $\ell_p$ Regression}
\author{
Anirban Dasgupta
\thanks{ Yahoo! Research, 701 First Ave., Sunnyvale, CA 94089.
Email: \{anirban, ravikumar, mahoney\}@yahoo-inc.com }
\and
Petros Drineas
\thanks{ Computer Science, Rensselaer Polytechnic Institute, Troy, NY 12180.
Work done while the author was visiting Yahoo! Research.
Email: drinep@cs.rpi.edu }
\and
Boulos Harb
\thanks{ Computer Science, University of Pennsylvania,
Philadelphia, PA 19107.
Work done while the author was visiting Yahoo! Research.
Email: boulos@cis.upenn.edu }
\and Ravi Kumar
\footnotemark[1]
\and Michael W. Mahoney
\footnotemark[1]
}

\date{}
\maketitle

\begin{abstract}
The \emph{$\ell_p$ regression problem} takes as input a matrix 
$A \in \Real^{n \times d}$, a vector $b \in \Real^n$, and a number 
$p \in [1,\infty)$, and it returns as output a number ${\cal Z}$ and a vector
$x_{\OPT} \in \Real^d$ such that
${\cal Z} = \min_{x \in \Real^d} \norm{{\A}x -b}_p = \norm{{\A}x_{\OPT}-b}_p$.
In this paper, we construct coresets and obtain an efficient two-stage
sampling-based approximation algorithm for the very overconstrained 
($n \gg d$) version of this classical problem, for all 
$\mbox{$p \in [1, \infty)$}$.
The first stage of our algorithm non-uniformly samples 
$\hat{r}_1 = O(36^p d^{\max\{p/2+1, p\}+1})$ rows of $A$ and the corresponding
elements of $b$, and then it solves the $\ell_p$ regression problem on the 
sample; we prove this is an $8$-approximation.
The second stage of our algorithm uses the output of the first stage to 
resample $\hat{r}_1/\epsilon^2$ constraints, and then it solves the
$\ell_p$ regression problem on the new sample; we prove this is a 
$(1+\epsilon)$-approximation.
Our algorithm unifies, improves upon, and extends the
existing algorithms for special cases of $\ell_p$ regression, namely
$p = 1,2$~\cite{clarksonL1, DM06}.
In course of proving our result, we develop two concepts---well-conditioned 
bases and subspace-preserving sampling---that 
are of independent interest.
\end{abstract}

\thispagestyle{empty}
\setcounter{page}{0}

\newpage

\section{Introduction}
\label{sxn:intro}
\vspace{-0.5cm}

An important question in algorithmic problem solving is whether there exists
a \emph{small} subset of the
input such that if computations are performed only on this subset, then
the solution to the given problem can be \emph{approximated} well.
Such a subset is often known as a \emph{coreset} for the problem.
The concept of coresets has
been extensively used in solving many problems in
optimization and computational geometry;
e.g., see the excellent survey by Agarwal, Har-Peled, and
Varadarajan \cite{coresetsurvey}.

In this paper, we construct coresets and obtain efficient sampling algorithms
for the classical $\ell_p$ regression problem, for all 
$\mbox{$p \in [1, \infty)$}$.
Recall the \emph{$\ell_p$ regression problem}:
\vspace{-0.25cm}
\begin{problem}[\emph{$\ell_p$ regression problem}]
\label{prob:intro:lpregression}
Let $\norm{\cdot}_p$ denote the $p$-norm of a vector.
Given as input a matrix ${\A} \in \Real^{n \times m}$, a target vector
$b \in \Real^n$, and a real number $p \in [1,\infty)$, find a vector
$x_{\OPT}$ and a number ${\cal Z}$ such that
\begin{equation}
\label{eqn:original_prob}
{\cal Z} = \min_{x \in \Real^m} \norm{{\A}x -b}_p
         = \norm{{\A}x_{\OPT} -b}_p  .
\end{equation}
\end{problem}
\noindent
In this paper, we will use the following \emph{$\ell_p$ regression coreset} 
concept:
\vspace{-0.25cm}
\begin{definition}[\emph{$\ell_p$ regression coreset}]
Let $0 < \epsilon < 1$.  
A \emph{coreset} for Problem~\ref{prob:intro:lpregression} is a set of 
indices $\mathcal{I}$ such that  the solution $\hat{x}_{\OPT}$ to
$\min_{x\in\Real^m}\snorm{\hat{A}x - \hat{b}}_p$,
where $\hat{A}$ is composed of those rows of $A$ whose indices are in
$\mathcal{I}$ and $\hat{b}$ consists of the corresponding elements of $b$,
satisfies $\norm{A\hat{x}_{\OPT} - b}_p \le (1+\epsilon)\min_x\norm{Ax - b}_p$.
\end{definition}
\noindent
If $n \gg m$, i.e., if there are many more constraints than variables, then 
(\ref{eqn:original_prob}) is an {\em overconstrained $\ell_p$ regression 
problem}. In this case, there does not in general exist a vector $x$ 
such that ${\A}x=b$, and thus ${\cal Z} > 0$. Overconstrained
regression problems are fundamental in statistical data analysis and
have numerous applications in applied mathematics, data mining, and
machine learning \cite{hast-tibs-fried,regression-analysis}.
Even though convex programming methods can be used to solve the
overconstrained regression problem in time $O((mn)^c)$, for $c > 1$,
this is prohibitive if $n$ is large.%
\footnote{For the special case of $p=2$, vector space methods can solve the 
regression problem in time $O(m^2n)$, and if $p=1$ linear programming methods 
can be used.}
This raises the natural
question of developing more efficient algorithms that run in time
$O(m^cn)$, for $c>1$, while possibly relaxing the solution to
Equation (\ref{eqn:original_prob}).
In particular: Can we get a $\kappa$-approximation to the $\ell_p$
regression problem, i.e., a vector $\hat{x}$ such that
$\norm{{\A}\hat{x} - b}_p \leq \kappa {\cal Z}$, where $\kappa > 1$?
Note that a coreset of small size would strongly satisfy our requirements 
and result in an efficiently computed solution that's almost as good as
the optimal.
Thus, the question becomes: Do coresets exist for the $\ell_p$ regression 
problem, and if so can we compute them efficiently?

Our main result is an efficient two-stage sampling-based approximation
algorithm that constructs a coreset and thus achieves a 
$(1+\epsilon)$-approximation for the $\ell_p$ regression problem.
The first-stage of the
algorithm is sufficient to obtain a (fixed) constant factor approximation.
The second-stage of the algorithm carefully uses the output of the 
first-stage to construct a coreset and achieve arbitrary constant factor
approximation.

\vspace{-0.50cm}
\subsection{Our contributions}
\label{sxn:intro:main_results}
\vspace{-0.25cm}

\para{Summary of results}
For simplicity of presentation, we summarize the results for the
case of $m = d = \mbox{rank}(A)$. Let $k = \max\{p/2+1,\ p\}$ and
let $\phi(r,d)$ be the time required to solve an $\ell_p$ regression
problem with $r$ constraints and $d$ variables. In the first stage
of the algorithm, we compute a set of sampling probabilities $p_1,
\ldots, p_n$ in time $O(nd^5 \log n)$, sample $\widehat{r_1}=O(36^p
d^{k+1})$ rows of ${\A}$ and the corresponding elements of $b$
according to the $p_i$'s, and solve an $\ell_p$ regression problem
on the (much smaller) sample; we prove this is an $8$-approximation algorithm 
with a running time of $O\left(nd^5\log n+\phi(\widehat{r_1},d)\right)$. In
the second stage of the algorithm, we use the residual from the
first stage to compute a new set of sampling probabilities $q_1,
\ldots, q_n$, sample additional $\widehat{r_2} =O(
\widehat{r_1}/\epsilon^2)$ rows of ${\A}$ and the corresponding
elements of $b$ according to the $q_i$'s, and solve an $\ell_p$
regression problem on the (much smaller) sample; we prove this is a
$(1+\epsilon)$-approximation algorithm with a total running time of
$O\left(nd^5\log n+\phi(\widehat{r_2},d)\right)$ (Section
\ref{sxn:main_algorithm}). 
We also show how to extend our basic algorithm to
commonly encountered and more general settings of constrained,
generalized, and weighted $\ell_p$ regression problems (Section
\ref{sxn:extensions}).

We note that the $l_p$ regression problem for $p=1,2$ has been 
studied before. For $p=1$, Clarkson~\cite{clarksonL1} uses a 
subgradient based algorithm to preprocess $A$ and $b$ and then 
samples the rows of the modified problem; these elegant techniques 
however depend crucially on the linear structure of the $l_1$ 
regression problem\footnote{Two ingredients of \cite{clarksonL1} use 
the linear structure: the subgradient based preprocessing itself, 
and the counting argument for the concentration bound.}.  
Furthermore, this algorithm does not yield coresets. For $p = 2$, 
Drineas, Mahoney, and Muthukrishnan~\cite{DM06} construct coresets 
by exploiting the singular value decomposition, a property peculiar 
to the $l_2$ space. Thus in order to efficiently compute coresets 
for the $\ell_p$ regression problem for all $p\in [1,\infty)$, we 
need tools that capture the geometry of $l_p$ norms. In this paper 
we develop the following two tools that may be of independent
interest (Section \ref{sxn:main_technical}).
 

\smallskip
\noindent (1) \emph{Well-conditioned bases.}\ Informally speaking,
if $U$ is a well-conditioned basis, then for all $z \in \Real^d$,
$\norm{z}_p$ should be close to $\norm{\U z}_p$.
We will formalize this by requiring that for all $z\in\Real^d$, $\norm{z}_q$ 
multiplicatively approximates $\norm{Uz}_p$ by a factor that can depend on $d$ 
but is \emph{independent} of $n$ (where $p$ and $q$ are conjugate; i.e., 
$q = p/(p-1)$).
We show that these bases exist and can be constructed in time $O(nd^5\log n)$. 
In fact, our notion of a well-conditioned basis can
be interpreted as a computational analog of the Auerbach and Lewis
bases studied in functional analysis \cite{bsa}.  
They are also related to the barycentric spanners recently introduced by 
Awerbuch and R. Kleinberg \cite{kleinberg-barycentric}
(Section~\ref{sxn:auerbach}). 
J. Kleinberg and Sandler~\cite{KS04} defined the notion of an 
$\ell_1$-independent basis, and our well-conditioned basis can be used to 
obtain an exponentially better ``condition number'' than their construction.
Further, Clarkson~\cite{clarksonL1} defined the notion of an 
``$\ell_1$-conditioned matrix,'' and he preprocessed the input matrix to an 
$\ell_1$ regression problem so that it satisfies conditions similar to those 
satisfied by our bases.

\smallskip 
\noindent (2) \emph{Subspace-preserving sampling.}\ We show that sampling rows
of ${\A}$ according to information in the rows of a well-conditioned
basis of $A$ minimizes the sampling variance and consequently, the
rank of $A$ is not lost by sampling. This is critical for our
relative-error approximation guarantees. The notion of
subspace-preserving sampling was used in~\cite{DM06} for $p=2$, 
but we abstract and generalize this concept for all $p\in[1,\infty)$.

We note that for $p=2$, our sampling complexity matches that 
of~\cite{DM06}, which is $O(d^2 /\epsilon^2)$; and for $p = 1$, it 
improves that of~\cite{clarksonL1} from $O(d^{3.5}(\log d 
)/\epsilon^2)$ to $O(d^{2.5}/\epsilon^2)$.

\para{Overview of our methods}
Given an input matrix $A$, we first construct a well-conditioned basis
for $A$ and use that to obtain bounds on a slightly non-standard
notion of a $p$-norm condition number of a matrix.
The use of this particular condition number is crucial since the variance
in the subspace preserving sampling can be upper bounded in terms of it.
An $\varepsilon$-net argument then shows that the first stage sampling gives
us a $8$-approximation.
The next twist is to use the output of the first stage as a feedback
to fine-tune the sampling probabilities.  This is done so that the
``positional information'' of $b$ with respect to $A$ is also
preserved in addition to the subspace.  A more careful use of a
different $\varepsilon$-net shows that the second
stage sampling achieves a $(1+\epsilon)$-approximation.


\vspace{-0.50cm}
\subsection{Related work}
\label{sxn:intro:prior_work}
\vspace{-0.25cm}

As mentioned earlier, in course of providing a sampling-based approximation 
algorithm for $\ell_1$ regression, Clarkson~\cite{clarksonL1} shows that 
coresets exist and can be computed efficiently for a \emph{controlled} 
$\ell_1$ regression problem. 
Clarkson first preprocesses the input matrix $A$ to make it well-conditioned 
with respect to the $\ell_1$ norm then applies a subgradient-descent-based 
approximation algorithm to guarantee that the $\ell_1$ norm of the target 
vector is conveniently bounded. 
Coresets of size $O(d^{3.5}\log d/\epsilon^2)$ are thereupon exhibited for 
this modified regression problem.
For the $\ell_2$ case, Drineas, Mahoney and Muthukrishnan~\cite{DM06}
designed sampling strategies to preserve the subspace
information of ${\A}$ and proved the existence of a coreset of rows of size
$O(d^2/\eps^2)$---for the \emph{original} $\ell_2$ regression problem; this 
leads to a $(1+\epsilon)$-approximation algorithm.
While their algorithm used $O(nd^2)$ time to construct the coreset and solve
the $\ell_2$ regression problem---which is 
sufficient time to solve the regression problem---in a subsequent work,
Sarl\'{o}s~\cite{Sarlos06} improved the running time for solving the 
regression problem to $\tilde{O}(nd)$ by using random sketches based on the 
Fast Johnson--Lindenstrauss transform of Ailon and Chazelle~\cite{AC06}.

%

More generally, embedding $d$-dimensional subspaces of $L_p$ into
$\ell_p^{f(d)}$ using coordinate restrictions has been extensively
studied \cite{schechtman,BLM89,talagrand90,talagrand95,Zavitch}.
Using well-conditioned bases, one can provide a constructive analog
of Schechtman's existential $L_1$ embedding result~\cite{schechtman}
(see also~\cite{BLM89}), that any $d$-dimensional subspace of
$L_1[0,1]$ can be embedded in $\ell_1^r$ with distortion
$(1+\epsilon)$ with $r = O(d^2/\epsilon^2)$, albeit with an extra
factor of $\sqrt{d}$ in the sampling complexity.
Coresets have been analyzed by the computation geometry community as a tool 
for efficiently approximating various extent 
measures~\cite{AHV04,coresetsurvey}; see 
also~\cite{coreset-kmeans,coreset-ball,FFS06} for applications of coresets in 
combinatorial optimization. An important difference is that most of
the coreset constructions are exponential in the dimension, and thus
applicable only to low-dimensional problems, whereas our coresets
are polynomial in the dimension, and thus applicable to
high-dimensional problems.

\vspace{-0.75cm}
\section{Preliminaries}
\label{sxn:review}
\vspace{-0.5cm}

Given a vector $x \in \Real^m$, its \emph{$p$-norm} is
$\norm{x}_p= \sum_{i=1}^{m} (\abs{x_i}^p)^{1/p}$, and the \emph{dual norm} of
$\norm{\cdot}_p$ is denoted $\norm{\cdot}_q$, where $1/p + 1/q =1$.
Given a matrix $A \in \Real^{n \times m}$, its \emph{generalized $p$-norm} is
$\mnorm{A}_p = (\sum_{i=1}^{n} \sum_{j=1}^{m} \abs{A_{ij}}^p)^{1/p}$.
This is a submultiplicative matrix norm that generalizes the Frobenius norm
from $p=2$ to all $p \in [1,\infty)$, but it is not a vector-induced matrix
norm.
The $j$-th column of $A$ is denoted $\col{A}{j}$, and the $i$-th row is
denoted $\row{A}{i}$.
In this notation,
$\mnorm{A}_p = (\sum_j\norm{\col{A}{j}}_p^p)^{1/p}
             = (\sum_i\norm{\row{A}{i}}_p^p)^{1/p}$.
For $x, x', x'' \in \Real^m$, it can be shown using H\"{o}lder's inequality
that
$\norm{x-x'}_p^p \le 2^{p-1} \left(\norm{x-x''}_p^p+\norm{x''-x'}_p^p\right)$.

Two crucial ingredients in our proofs are $\vareps$-nets and tail-inequalities.
A subset $\mathcal{N}(D)$ of a set $D$ is called an \emph{$\varepsilon$-net}
in $D$ for some $\varepsilon > 0$ if for every $x\in D$, there is a
$y\in \mathcal{N}(D)$ with $\norm{x-y} \le \varepsilon$.
In order to construct an $\vareps$-net for ${D}$ it is enough to choose
$\mathcal{N}(D)$ to be the maximal set of points that are pairwise $\vareps$
apart.
It is well known that the unit ball of a $d$-dimensional space has an
$\vareps$-net of size at most $(3/\vareps)^d$~\cite{BLM89}.

Finally, throughout this paper, we will use the following sampling matrix
formalism to represent our sampling operations.
Given a set of $n$ probabilities, $p_i \in (0,1]$, for $i=1,\ldots,n$, let
$S$ be an $n \times n$ diagonal sampling matrix such that $S_{ii}$ is set
to $1/p_i^{1/p}$ with probability $p_i$ and to zero otherwise.
Clearly, premultiplying ${\A}$ or $b$ by $S$ determines whether the $i$-th
row of ${\A}$ and the corresponding element of $b$ will be included in the
sample, and the expected number of rows/elements selected is
$r' = \sum_{i=1}^n p_i$.
(In what follows, we will abuse notation slightly by ignoring zeroed out rows
and regarding $S$ as an $r' \times n$ matrix and thus $SA$ as an
$r' \times m$ matrix.)
Thus, e.g., sampling constraints from Equation~\eqref{eqn:original_prob} and
solving the induced subproblem may be represented as solving
\begin{equation}
\label{sample_problem}
\hat{{\cal Z}} = \min_{\hat{x} \in \Real^m} \norm{S{\A}\hat{x}-Sb}_p .
\end{equation}
A vector $\hat{x}$ is said to be a \emph{$\kappa$-approximation} to the
$\ell_p$ regression problem of Equation~\eqref{eqn:original_prob}, for
$\kappa \geq 1$, if $\norm{{\A}\hat{x} - b}_p \leq \kappa {\cal Z}$.
Finally, the Appendix contains all the missing proofs.

\vspace{-0.5cm}
\section{Main technical ingredients}
\label{sxn:main_technical}


\vspace{-0.25cm}
\subsection{Well-conditioned bases}
\label{sxn:auerbach}
\vspace{-0.25cm}

We introduce the following notion of a ``well-conditioned'' basis.

\begin{definition}[Well-conditioned basis]
\label{def:good_basis}
Let ${\A}$ be an $n\times m$ matrix of rank $d$, let $p \in [1,\infty)$, and
let $q$ be its dual norm.
Then an $n \times d$ matrix ${\U}$ is an
\emph{$(\alpha,\beta,p)$-well-conditioned basis} for the column space of
${\A}$ if
(1) $\mnorm{{\U}}_p \le \alpha$, and
(2) for all $z\in\Real^d$, $\norm{z}_q \le \beta \norm{\U z}_p$.
We will say that ${\U}$ is a \emph{$p$-well-conditioned basis} for the column
space of ${\A}$ if $\alpha$ and $\beta$ are $d^{O(1)}$,
independent of $m$ and $n$.
\end{definition}

\noindent
Recall that any orthonormal basis ${\U}$ for $\spn({\A})$ satisfies both
$\mnorm{{\U}}_2 = \norm{{\U}}_F = \sqrt{d}$ and also
$\norm{z}_2 = \norm{\U z}_2$ for all $z\in \Real^d$, and thus is a
$(\sqrt{d},1,2)$-well-conditioned basis.
Thus, Definition~\ref{def:good_basis} generalizes to an arbitrary $p$-norm,
for $p \in [1,\infty)$, the notion that an orthogonal matrix is
well-conditioned with respect to the $2$-norm.
Note also that duality is incorporated into Definition~\ref{def:good_basis}
since it relates the $p$-norm of the vector $z\in \Real^d$ to the $q$-norm of
the vector $\U z\in \Real^n$, where $p$ and $q$ are dual.%
\footnote{For $p=2$, Drineas, Mahoney, and Muthukrishnan used this basis, 
i.e., an orthonormal matrix, to construct probabilities to sample the original 
matrix.  For $p=1$, Clarkson used a procedure similar to the one we describe 
in the proof of Theorem~\ref{thm:good_basis} to preprocess $A$ such that the 
$1$-norm of $z$ is a $d\sqrt{d}$ factor away from the $1$-norm of $Az$.}

The existence and efficient construction of these bases is given by
the following.

\begin{theorem}
\label{thm:good_basis}
Let ${\A}$ be an $n\times m$ matrix of rank $d$, let $p \in [1,\infty)$, and
let $q$ be its dual norm.
Then there exists an \emph{$(\alpha,\beta,p)$-well-conditioned basis} ${\U}$
for the column space of ${\A}$ such that:
if $p<2$, then $\alpha = d^{\frac1p + \frac12}$
                    and $\beta = 1$,
if $p=2$, then $\alpha = d^{\frac12}$
                    and $\beta = 1$, and
if $p>2$, then $\alpha = d^{\frac1p + \frac12}$
                    and $\beta = d^{\frac1q - \frac12}$.
Moreover, ${\U}$ can be computed in $O(nmd+nd^5\log n)$ time (or in just
$O(nmd)$ time if $p=2$).
\end{theorem}
\begin{proof}
Let ${\A}=QR$, where $Q$ is any $n\times d$ matrix that is an orthonormal
basis for $\spn({\A})$ and $R$ is a $d \times m$ matrix.
If $p=2$, then $Q$ is the desired basis ${\U}$; from the discussion following
Definition~\ref{def:good_basis}, $\alpha=\sqrt{d}$ and $\beta=1$,
and computing it requires $O(nmd)$ time.
Otherwise, fix $Q$ and $p$ and define the norm,
$ \norm{z}_{Q,p} \triangleq \norm{Qz}_p$ .
A quick check shows that $\norm{\cdot}_{Q,p}$ is indeed a norm.
   ($\norm{z}_{Q,p} = 0$ if and only if $z = 0$
   since $Q$ has full column rank;
   $\norm{\gamma z}_{Q,p} = \norm{\gamma Qz}_p
   = \abs{\gamma}\norm{Qz}_p = \abs{\gamma}\norm{z}_{Q,p}$; and
   $\norm{z + z'}_{Q,p} = \norm{Q(z+z')}_p
   \le \norm{Qz}_p + \norm{Qz'}_p = \norm{z}_{Q,p} + \norm{z'}_{Q,p}$.)

Consider the set $C = \{z \in \Real^d : \norm{z}_{Q,p} \le 1\}$, which is the
unit ball of the norm $\norm{\cdot}_{Q,p}$.
In addition, define the $d \times d$ matrix $F$ such that
$\mathcal{E}_{\text{\sc lj}} = \{z\in\Real^d: z^T F z \le 1\}$ is the
L\"{o}wner--John ellipsoid of $C$.
Since $C$ is symmetric about the origin,
$
(1/\sqrt{d}) \mathcal{E}_{\text{\sc lj}}
   \subseteq C
   \subseteq \mathcal{E}_{\text{\sc lj}}$;
thus, for all $z\in\Real^d$,
\begin{equation}
\label{eq-lj-bounds}
\norm{z}_{\text{\sc lj}}
   \le \norm{z}_{Q,p}
   \le \sqrt{d}\norm{z}_{\text{\sc lj}} \enspace ,
\end{equation}
where $\norm{z}_{\text{\sc lj}}^2 = z^T F z$
(see, e.g.~\cite[pp.~413--4]{Boyd04}).
Since the matrix $F$ is symmetric positive definite, we can express it as
$F = G^T G$, where $G$ is full rank and upper triangular.
Since $Q$ is an orthogonal basis for $\spn({\A})$ and $G$ is a $d\times d$
matrix of full rank, it follows that ${\U} = QG^{-1}$ is an $n\times d$ matrix
that spans the column space of ${\A}$.
We claim that ${\U} \triangleq QG^{-1}$ is the desired
$p$-well-conditioned basis.

To establish this claim, let $z' = Gz$.
Thus, $\norm{z}_{\text{\sc lj}}^2 = z^T Fz = z^TG^T Gz = (Gz)^T Gz = {z'}^T z' = \norm{z'}_2^2$.
Furthermore, since $G$ is invertible, $z = G^{-1} z'$, and thus
$\norm{z}_{Q,p} = \norm{Qz}_p = \norm{QG^{-1} z'}_p$.
By combining these expression with~(\ref{eq-lj-bounds}), it follows that
for all $z'\in\Real^d$,
\begin{equation}
\label{eq-l2-bounds}
\norm{z'}_2 \le \norm{\U z'}_p \le \sqrt{d}\norm{z'}_2 \enspace .
\end{equation}
Since
$\mnorm{\U}_p^p
   = \sum_j \norm{\col{\U}{j}}_p^p
   = \sum_j\norm{\U e_j}_p^p
   \le \sum_j d^{\frac p2}\norm{e_j}_2^p
   = d^{\frac p2+1}$,
 where the inequality follows from the upper bound in \eqref{eq-l2-bounds},
it follows that $\alpha = d^{\frac1p + \frac12}$.
If $p<2$, then $q > 2$ and $ \norm{z}_q \le \norm{z}_2 $ for all
$z\in \Real^d$;
by combining this with \eqref{eq-l2-bounds}, it follows that $\beta = 1$.
On the other hand, if $p >2$, then $q < 2$ and
$\norm{z}_q \le d^{\frac1q - \frac12}\norm{z}_2$;
by combining this with \eqref{eq-l2-bounds}, it follows that
$\beta = d^{\frac1q - \frac12}$.

In order to construct $\U$, we need to compute $Q$ and $G$ and then invert
$G$.
Our matrix ${\A}$ can be decomposed into $QR$ using the compact $QR$
decomposition in $O(nmd)$ time.
The matrix $F$ describing the L\"{o}wner--John ellipsoid of the unit ball of
$\norm{\cdot}_{Q,p}$ can be computed in $O(nd^5\log n)$ time.
Finally, computing $G$ from $F$ takes $O(d^3)$ time, and inverting $G$ takes
$O(d^3)$ time.
\end{proof}

\para{Connection to barycentric spanners}
A point set $K = \{K_1,\ldots, K_d\} \subseteq D \subseteq \Real^d$ is a
\emph{barycentric spanner} for the set $D$ if every $z \in D$ may be
expressed as a linear combination of elements of $K$ using
coefficients in $[-C,C]$, for $C = 1$. When $C > 1$, $K$ is called a
$C$-approximate barycentric spanner. Barycentric spanners were
introduced by Awerbuch and R. Kleinberg in
\cite{kleinberg-barycentric}. They showed that if a set is compact,
then it has a barycentric spanner. Our proof shows that if ${\A}$ is
an $n\times d$ matrix, then $\tau^{-1} = R^{-1}G^{-1}
\in\Real^{d\times d}$ is a $\sqrt{d}$-approximate barycentric spanner
for $D = \{z\in \Real^d : \norm{{\A} z}_p \le 1\}$.  To see this,
first note that each $\col{\tau}{j}^{-1}$ belongs to $D$ since
$\|\A\col{\tau}{j}^{-1}\|_p = \norm{\U e_j}_p \le \norm{e_j}_2 =1$,
where the inequality is obtained from Equation~\eqref{eq-l2-bounds}.
Moreover, since $\tau^{-1}$ spans $\Real^d$, we can write any $z\in D$
as $z = \tau^{-1}\nu$.
Hence,
\[ \frac{\norm{\nu}_\infty}{\sqrt{d}}
  \le \frac{\norm{\nu}_2}{\sqrt{d}}
  \le \norm{U\nu}_p = \norm{\A\tau^{-1}\nu}_p
  =   \norm{\A z}_p
  \le 1   \enspace ,
\]
where the second inequality is also obtained from
Equation~\eqref{eq-l2-bounds}.
This
shows that our basis has the added property that every element $z\in
D$ can be expressed as a linear combination of elements (or columns)
of $\tau^{-1}$ using coefficients whose {$\ell_2$ norm is
bounded by $\sqrt{d}$.

\para{Connection to Auerbach bases}
An \emph{Auerbach basis} $U = \{\col{\U}{j}\}_{j=1}^d$ for a
$d$-dimensional normed space $\mathcal{A}$ is a basis such that
$\norm{\col{\U}{j}}_p = 1$ for all $j$ and such that whenever $y =
\sum_j \nu_j\col{\U}{j}$ is in the unit ball of $\mathcal{A}$ then
$\abs{\nu_j} \le 1$.
The existence of such a basis for every finite
dimensional normed space was first proved by Herman Auerbach
\cite{auerbach} (see also~\cite{MDay,ATaylor}).
It can easily be
shown that an Auerbach basis is an $(\alpha,\beta,p)$-well-conditioned
basis, with $\alpha = d$ and $\beta = 1$ for all $p$.  Further, suppose
$\U$ is an Auerbach basis for $\spn(\A)$, where $\A$ is an $n\times d$
matrix of rank $d$.
Writing $\A = \U\tau$, it follows that $\tau^{-1}$ is an \emph{exact}
barycentric spanner for $D = \{z\in \Real^d : \norm{{\A} z}_p \le 1\}$.
Specifically, each
$\col{\tau}{j}^{-1} \in D$ since $\|\A\col{\tau}{j}^{-1}\|_p =
\norm{\col{U}{j}}_p =1$. Now write $z\in D$ as $z = \tau^{-1}\nu$.
Since the vector $y = \A z = \U\nu$ is in the unit ball of $\spn(\A)$,
we have $\abs{\nu_j} \le 1$ for all $1\le j \le d$.  Therefore,
computing a barycentric spanner for the compact set $D$---which is the
pre-image of the unit ball of $\spn(A)$---is equivalent (up to
polynomial factors) to computing an Auerbach basis for $\spn(A)$.

\vspace{-0.50cm}
\subsection{Subspace-preserving sampling}
\label{sxn:main_technical:subspace}
\vspace{-0.25cm}

In the previous subsection
(and in the notation of the proof of Theorem~\ref{thm:good_basis}),
we saw that given $p \in [1,\infty)$, any
$n \times m$ matrix ${\A}$ of rank $d$ can be decomposed as
\[ \A = QR = QG^{-1}GR = \U\tau \enspace , \]
where $\U = QG^{-1}$ is a $p$-well-conditioned basis for $\spn(\A)$ and
$\tau = GR$.
The significance of a $p$-well-conditioned basis is that we are able to
minimize the variance in our sampling process by randomly sampling \emph{rows}
of the matrix ${\A}$ and elements of the vector $b$ according to a probability
distribution that depends on norms of the \emph{rows} of the matrix ${\U}$.
This will allow us to preserve the subspace structure of $\spn(\A)$ and thus
to achieve relative-error approximation guarantees.

More precisely, given $p \in [1,\infty)$ and any $n \times m$ matrix ${\A}$ of
rank $d$ decomposed as ${\A} = {\U}\tau$, where ${\U}$ is an
$(\alpha,\beta,p)$-well-conditioned basis for $\spn(\A)$, consider any set of
sampling probabilities $p_i$ for $i=1,\ldots,n$, that satisfy:
\begin{align}
\label{eq-p-subspace}
p_i & \ge \min \left\{1, \frac{\norm{\row{{\U}}{i}}_p^p}{\mnorm{{\U}}_p^p}r \right\}
   \enspace ,
\end{align}
where $r = r(\alpha, \beta, p, d,\epsilon)$ to be determined below.
Let us randomly sample the $i^{th}$ row of ${\A}$ with probability
$p_i$, for all $i=1,\ldots,n$. Recall that we can construct a
diagonal sampling matrix $S$, where each $S_{ii}= 1/p_i^{1/p}$ with
probability $p_i$ and $0$ otherwise, in which case we can represent
the sampling operation as $S{\A}$.

The following theorem is our main result regarding this subspace-preserving
sampling procedure.

\begin{theorem}
\label{thm:preserve_subspace}
Let ${\A}$ be an $n\times m$ matrix of rank $d$, and let $p \in [1,\infty)$.
Let ${\U}$ be an $(\alpha,\beta,p)$-well-conditioned basis for $\spn(\A)$,
and let us randomly sample rows of ${\A}$ according to the procedure described
above using the probability distribution given by
Equation~\eqref{eq-p-subspace}, where
$r \ge 32^p(\alpha\beta)^p(d\ln(\frac{12}{\eps})+\ln(\frac{2}{\delta}))/(p^2\eps^{2})$.
Then, with probability $1-\delta$, the following holds for all
$x \in \Real^m$:
$$
\abs{\;\norm{S \A x}_p-\norm{\A x}_p\;} \le \epsilon\norm{\A x}_p  .
$$
\end{theorem}

Several things should be noted about this result.
First, it implies that $\mbox{rank}(S \A) = \mbox{rank}(\A)$, since otherwise
we could choose a vector $x \in \mbox{null}(S \A)$ and violate the
theorem.
In this sense, this theorem generalizes the subspace-preservation result of
Lemma $4.1$ of~\cite{DM06} to all $p \in [1,\infty)$.
Second, regarding sampling complexity:
if $p<2$ the sampling complexity is $O(d^{\frac{p}{2}+2})$,
if $p=2$ it is $O(d^2)$, and if $p >2$ it is
$O(dd^{\frac1p+\frac12}d^{\frac1q-\frac12})^p = O(d^{p+1})$.
Finally, note that this theorem is analogous to the main result of
Schechtman~\cite{schechtman}, which uses the notion of Auerbach bases.

\vspace{-0.5cm}
\section{The sampling algorithm}
\label{sxn:main_algorithm}

\vspace{-0.25cm}
\subsection{Statement of our main algorithm and theorem}
\label{sxn:main_algorithm:statement}
\vspace{-0.25cm}

Our main sampling algorithm for approximating the solution to the $\ell_p$
regression problem is presented in Figure~\ref{alg:main}.%
\footnote{
It has been brought to our attention by an anonymous reviewer that one of the main results of this section can be obtained with a simpler analysis.  In particular, one can show that one can obtain a relative error (as opposed to a constant factor) approximation in one stage, if the sampling probabilities are constructed from subspace information in the augmented matrix $[A b]$ (as opposed to using just subspace information from the matrix $A$), i.e., by using information in both the data matrix $A$ and the target vector $b$.
}
The algorithm takes as input an $n \times m$ matrix ${\A}$ of rank $d$, a
vector $b\in\Real^n$, and a number $p \in [1,\infty)$.
It is a two-stage algorithm that returns as output a vector
$\hat{x}_{\OPT} \in \Real^m$ (or a vector $\hat{x}_c \in \Real^m$ if only the
first stage is run).
In either case, the output is the solution to the induced $\ell_p$ regression
subproblem constructed on the randomly sampled constraints.

\begin{figure}[ht]
\framebox[6.5in]{\parbox{6.3in}{

\textbf{Input:}
An $n \times m$ matrix ${\A}$ of rank $d$, a vector $b\in\Real^n$, and
$p \in [1,\infty)$.
\medskip

\noindent Let $0 < \epsilon < 1/7$, and define $k =
\max\{p/2+1,p\}$.
%
\begin{itemize}
\item[-]
Find a $p$-well-conditioned basis ${\U} \in \Real^{n\times d}$ for
$\spn({\A})$ (as in the proof of Theorem~\ref{thm:good_basis}) .

\item[-]
\textbf{Stage 1:} Define $p_i = \min \left\{1,
\frac{\norm{\row{\U}{i}}_p^p}{\mnorm{\U}_p^p}r_1 \right\}$ where
$r_1 =
 8^2\cdot 36^p d^k\left(d\ln(8\cdot 36) + \ln(200)\right)$.
\begin{itemize}
\item[-]
Generate (implicitly) $\Sa$ where $\Sa_{ii} = 1/p_i^{1/p}$ with probability
$p_i$ and $0$ otherwise.
\item[-]
Let $\hat{x}_c$ be the solution to
$\displaystyle\min_{x\in\Real^m}\norm{\Sa(\A x-b)}_p$.
\end{itemize}

\item[-]
\textbf{Stage 2:} Let $\hat{\rho} = {\A}\hat{x}_c - b$, and unless
$\hat{\rho} = 0$ define $q_i = \min \left\{1, \max
\left\{p_i,\frac{\abs{\hat{\rho}_i}^p}{\|\hat{\rho}\|_p^p}r_2
\right\} \right\}$ with $r_2 = \frac{36^p
d^k}{\epsilon^2}\left(d\ln(\frac{36}{\epsilon}) + \ln(200)\right)$.
\begin{itemize}
\item[-]
Generate (implicitly, a new) $\Sb$ where $\Sb_{ii} = 1/q_i^{1/p}$ with
probability $q_i$ and $0$ otherwise.
\item[-]
Let $\hat{x}_{\OPT}$ be the solution to
$\displaystyle\min_{x\in\Real^m}\norm{\Sb(\A x-b)}_p$.
\end{itemize}
\end{itemize}

\noindent
\textbf{Output:}
$\hat{x}_{\OPT} $ (or $\hat{x}_c$ if only the first stage is run). }}
\caption{Sampling algorithm for $\ell_p$ regression.}
\label{alg:main}
\end{figure}
\medskip

The algorithm first computes a $p$-well-conditioned basis ${\U}$ for
$\spn({\A})$, as described in the proof of
Theorem~\ref{thm:good_basis}. Then, in the first stage, the
algorithm uses information from the norms of the rows of ${\U}$ to
sample constraints from the input $\ell_p$ regression problem. In
particular, roughly $O(d^{p+1})$ rows of ${\A}$, and the
corresponding elements of $b$, are randomly sampled according to the
probability distribution given by 
 \begin{align}
\label{eq-p} p_i & = \min \left\{1,
\frac{\norm{\row{{\U}}{i}}_p^p}{\mnorm{{\U}}_p^p}r_1 \right\},
\mbox{ where } r_1 =8^2\cdot 36^p d^k\left(d\ln(8\cdot 36) +
\ln(200)\right).
\end{align}
implicitly represented by a diagonal sampling matrix $S$, where
each $S_{ii} = 1/p_i^{1/p}$.
For the remainder of the paper, we will use $\Sa$ to denote the sampling
matrix for the first-stage sampling probabilities.
The algorithm then solves, using any $\ell_p$ solver of one's choice, the 
smaller subproblem.
If the solution to the induced subproblem is denoted $\hat{x}_c$, then, as we
will see in Theorem~\ref{thm-main}, this is an $8$-approximation to the
original problem.%
\footnote{For $p=2$, Drineas, Mahoney, and Muthhukrishnan show that this 
first stage actually leads to a $(1+\epsilon)$-approximation.  For $p=1$, 
Clarkson develops a subgradient-based algorithm and runs it, after 
preprocessing the input, on all the input constraints to obtain a 
constant-factor approximation in a stage analogous to our first stage.  Here, 
however, we solve an $\ell_p$ regression problem on a small subset of the 
constraints to obtain the constant-factor approximation.  Moreover, our 
procedure works for all $p \in [1,\infty)$.}

In the second stage, the algorithm uses information from the residual of the
$8$-approximation computed in the first stage to refine the sampling
probabilities.
Define the residual $\hat{\rho} = {\A}\hat{x}_{c} - b$ (and note that
$\|\hat{\rho}\|_p \leq 8\,{\cal Z}$).
Then, roughly $O(d^{p+1}/\epsilon^2)$ rows of ${\A}$, and the corresponding
elements of $b$, are randomly sampled according to the probability distribution
\begin{align}
\label{eq-q} q_i = \min \left\{1, \max
\left\{p_i,\frac{\abs{\hat{\rho}_i}^p}{\|\hat{\rho}\|_p^p}r_2
\right\} \right\}, \mbox{where } r_2 = \frac{36^p
d^k}{\epsilon^2}\left(d\ln(\frac{36}{\epsilon}) + \ln(200)\right).
\end{align}
As before, this can be represented as a diagonal sampling matrix $\Sb$, where
each $\Sb_{ii}= 1/q_i^{1/p}$ with probability $q_i$ and $0$ otherwise.
For the remainder of the paper, we will use $\Sb$ to denote the sampling
matrix for the second-stage sampling probabilities.
Again, the algorithm solves, using any $\ell_p$ solver of one's choice, the 
smaller subproblem.
If the solution to the induced subproblem at the second stage is denoted
$\hat{x}_{\OPT} $, then, as we will see in Theorem~\ref{thm-main}, this is
a $(1+\eps)$-approximation to the original problem.%
\footnote{The subspace-based sampling probabilities (\ref{eq-p}) are similar 
to those used by Drineas, Mahoney, and Muthukrishnan~\cite{DM06}, while
the residual-based sampling probabilities (\ref{eq-q}) are similar to those 
used by Clarkson~\cite{clarksonL1}.}

The following is our main theorem for the $\ell_p$
regression algorithm presented in Figure~\ref{alg:main}.

\begin{theorem}
\label{thm-main} Let ${\A}$ be an $n\times m$ matrix of rank $d$,
let $b \in \Real^n$, and let $p \in [1,\infty)$. Recall that $r_1 =
8^2\cdot 36^p d^k\left(d\ln(8\cdot 36) + \ln(200)\right)$ and $r_2
=\frac{36^p d^k}{\epsilon^2}\left(d\ln(\frac{36}{\epsilon}) +
\ln(200)\right) $.
Then,
\begin{itemize}
\item
\textbf{Constant-factor approximation.}
If only the first stage of the algorithm in Figure \ref{alg:main}
is run, then with probability at least $0.6$,
the solution $\hat{x}_c$ to the sampled problem based on the $p_i$'s of
Equation~\eqref{eq-p-subspace} is an $8$-approximation to the
$\ell_p$ regression problem;
\item
\textbf{Relative-error approximation.}
If both stages of the algorithm are run, then with probability at least $0.5$,
the solution $\hat{x}_{\OPT}$ to the sampled problem based on the $q_i$'s of
Equation~\eqref{eq-q} is a $(1+\epsilon)$-approximation to the
$\ell_p$ regression problem;
\item
\textbf{Running time.} The $i^{th}$ stage of the algorithm runs in
time $O(nmd + nd^5\log n + \phi(20 i r_i,m))$, where $\phi(s,t)$ is the
time taken to solve the regression problem
$\min_{x\in\Real^t}\norm{{\A}'x-b'}_p$, where ${\A}'\in
\Real^{s\times t}$ is of rank $d$ and $b'\in\Real^s$.
\end{itemize}
\end{theorem}

\noindent
Note that since the algorithm of Figure~\ref{alg:main} constructs the
$(\alpha,\beta,p)$-well-conditioned basis ${\U}$ using the procedure in the
proof of Theorem~\ref{thm:good_basis}, our sampling complexity depends on
$\alpha$ and $\beta$.
In particular, it will be $O(d(\alpha\beta)^p)$.
Thus, if $p<2$ our sampling complexity is
$O(d\cdot d^{\frac{p}{2} + 1}) = O(d^{\frac{p}{2} + 2})$;
if $p >2$ it is
$O(d(d^{\frac1p+\frac12}d^{\frac1q-\frac12})^p) = O(d^{p+1})$; and
(although not explicitly stated, our proof will make it clear that)
if $p=2$ it is $O(d^2)$.
Note also that we have stated the claims of the theorem as holding with
constant probability, but they can be shown to hold with probability at least
$1-\delta$ by using standard amplification techniques.

\vspace{-0.25cm}
\subsection{Proof for first-stage sampling -- constant-factor approximation}
\label{sxn:main_algorithm:proof1}
\vspace{-0.25cm}

To prove the claims of Theorem~\ref{thm-main} having to do with the
output of the algorithm after the first stage of sampling, we begin
with two lemmas. First note that, because of our choice of $r_1$, we
can use the subspace preserving Theorem \ref{thm:preserve_subspace}
with only a constant distortion, i.e., for all $x$, we have
\begin{align*}
\frac{7}{8}\norm{Ax}_p\leq \norm{\Sa Ax}_p \le
\frac{9}{8}\norm{Ax}_p
\end{align*}
with probability at least $0.99$. The first lemma below now states
that the optimal solution to the original problem provides a small
(constant-factor) residual when evaluated in the sampled problem.

\begin{lemma}
\label{lemma-perturb3Opt}
$\norm{\Sa(\A x_{\OPT} - b)} \le 3{\cal Z}$, with probability at least
$1-1/3^p$.
\end{lemma}

The next lemma states that if the solution to the sampled problem provides a
constant-factor approximation (when evaluated in the sampled problem), then
when this solution is evaluated in the original regression problem we get a
(slightly weaker) constant-factor approximation.

\begin{lemma}
\label{lemma-constantOpt}
If $\norm{\Sa({\A}\hat{x}_c - b)} \le 3\,{\cal Z}$, then
$\norm{{\A}\hat{x}_c - b} \le 8\,{\cal Z}$.
\end{lemma}

Clearly, $\norm{ \Sa({\A}\hat{x}_c - b)} \le \norm{ \Sa({\A}x_{\OPT} - b)}$
(since $\hat{x}_c$ is an optimum for the sampled $\ell_p$ regression problem).
Combining this with Lemmas~\ref{lemma-perturb3Opt}
and~\ref{lemma-constantOpt}, it follows that the solution $\hat{x}_c$ to the
the sampled problem based on the $p_i$'s of Equation~\eqref{eq-p-subspace}
satisfies
$\norm{{\A}\hat{x}_c - b} \leq 8\,{\cal Z}$, i.e., $\hat{x}_c$ is an
$8$-approximation to the original ${\cal Z}$.

To conclude the proof of the claims for the first stage of sampling,
note that by our choice of $r_1$,
Theorem~\ref{thm:preserve_subspace} fails to hold for our first
stage sampling with probability no greater than $1/100$. In
addition, Lemma~\ref{lemma-perturb3Opt} fails to hold with
probability no grater than $1/3^p$, which is no greater than $1/3$
for all $p \in [1,\infty)$. 
Finally, let $\widehat{r_1}$ be a random variable representing the number of 
rows actually chosen by our sampling schema, and note that 
$ \Ex{\widehat{r_1}} \le r_1$. 
By Markov's inequality, it follows that $\widehat{r_1} > 20r_1$ with
probability less than $1/20$. Thus, the first stage of our algorithm
fails to give an $8$-approximation in the specified running time
with a probability bounded by $1/3+1/20+1/100 < 2/5$.

\vspace{-0.25cm}
\subsection{Proof for second-stage sampling -- relative-error approximation}
\label{sxn:main_algorithm:proof2}
\vspace{-0.25cm}

The proof of the claims of Theorem~\ref{thm-main} having to do with
the output of the algorithm after the second stage of sampling will
parallel that for the first stage, but it will have several
technical complexities that arise since the first triangle
inequality approximation in the proof of
Lemma~\ref{lemma-constantOpt} is too coarse for relative-error
approximation. By our choice of $r_2$ again, we have a finer result
for subspace preservation. Thus, with probability $0.99$, the
following holds for all $x$
\begin{align*}
(1 - \epsilon) \norm{Ax}_p \le \norm{SAx}_p \le (1 + \epsilon)
\norm{Ax}_p
\end{align*}
As before, we start with a lemma that states that the optimal
solution to the original problem provides a small (now a
relative-error) residual when evaluated in the sampled problem. This
is the analog of Lemma~\ref{lemma-perturb3Opt}. An important
difference is that the second stage sampling probabilities
significantly enhance the probability of success.

\begin{lemma}
\label{lemma-perturbOpt}
$\norm{\Sb(\A {x}_{\OPT} - b)}\le (1 + \epsilon){\cal Z}$, with probability at
least $0.99$.
\end{lemma}

Next we show that if the solution to the sampled problem
provides a relative-error approximation (when evaluated in the sampled
problem), then when this solution is evaluated in the original regression
problem we get a (slightly weaker) relative-error approximation.
We first establish two technical lemmas.

The following lemma says that for all optimal solutions $\hat{x}_{\OPT}$ to
the second-stage sampled problem, ${\A}\hat{x}_{\OPT}$ is not too far from
${\A}\hat{x}_c$, where $\hat{x}_c$ is the optimal solution from the first
stage, in a $p$-norm sense.
Hence, the lemma will allow us to restrict our calculations in
Lemmas~\ref{lemma-optnet2} and~\ref{lemma-nn} to the ball of radius
$12\,{\cal Z}$ centered at ${\A}\hat{x}_c$.

\begin{lemma}
\label{lemma-3opt}
$\norm{{\A}\hat{x}_{\OPT} - {\A}\hat{x}_c} \le 12\,{\cal Z}$.
\end{lemma}

Thus, if we define the affine ball of radius $12\,{\cal Z}$ that is centered
at ${\A}\hat{x}_c$ and that lies in $\spn({\A})$,
\begin{equation}
\label{eq-affineball}
B = \{y\in\Real^n : y = \A x,
                    x\in\Real^m,
                    \norm{{\A}\hat{x}_c - y} \le 12\,{\cal Z}\} \enspace ,
\end{equation}
then Lemma~\ref{lemma-3opt} states that ${\A}\hat{x}_{\OPT} \in B$, for all
optimal solutions $\hat{x}_{\OPT}$ to the sampled problem.
Let us consider an $\varepsilon$-net, call it $B_\varepsilon$, with
$\varepsilon = \epsilon\,{\cal Z}$, for this ball $B$.
Using standard arguments, the size of the $\varepsilon$-net is
$\left(\frac{3\cdot 12\,{\cal Z}}{\epsilon\,{\cal Z}}\right)^d
   = \left(\frac{36}{\epsilon}\right)^d$.
The next lemma states that for all points in the $\varepsilon$-net, if that
point provides a relative-error approximation (when evaluated in the sampled
problem), then when this point is evaluated in the original regression
problem we get a (slightly weaker) relative-error approximation.

\begin{lemma}
\label{lemma-optnet2}
For all points $\A x_\varepsilon$ in the $\varepsilon$-net, $B_\varepsilon$,
if $\norm{\Sb(\A x_\varepsilon - b)} \le (1+3\epsilon){\cal Z}$, then
$\norm{\A x_\varepsilon - b} \le (1+6\epsilon){\cal Z}$, with probability
$0.99$.
\end{lemma}

Finally, the next lemma states that if the solution to the sampled problem
(in the second stage of sampling) provides a relative-error approximation
(when evaluated in the sampled problem), then when this solution is evaluated
in the original regression problem we get a (slightly weaker) relative-error
approximation.
This is the analog of Lemma~\ref{lemma-constantOpt}, and its proof will use
Lemma~\ref{lemma-optnet2}.

\begin{lemma}
\label{lemma-nn}
If $\norm{\Sb({\A}\hat{x}_{\OPT} - b)} \le (1+\epsilon){\cal Z}$,
then $\norm{{\A}\hat{x}_{\OPT} - b} \le (1+7\epsilon){\cal Z}$.
\end{lemma}

Clearly,
$\norm{ \Sb({\A}\hat{x}_{\OPT} - b)} \le \norm{ \Sb({\A}x_{\OPT} - b)}$,
since $\hat{x}_{\OPT}$ is an optimum for the sampled $\ell_p$ regression
problem.
Combining this with Lemmas~\ref{lemma-perturbOpt} and~\ref{lemma-nn}, it
follows that the solution $\hat{x}_{\OPT}$ to the the sampled problem based on
the $q_i$'s of Equation~\eqref{eq-q} satisfies
$\norm{{\A}\hat{x}_{\OPT} - b} \leq (1+\epsilon)\,{\cal Z}$, i.e.,
$\hat{x}_{\OPT}$ is a $(1+\epsilon)$-approximation to the original ${\cal Z}$.

To conclude the proof of the claims for the second stage of sampling, recall 
that the first stage failed with probability no greater than $2/5$. 
Note also that by our choice of $r_2$, Theorem~\ref{thm:preserve_subspace} 
fails to hold for our second stage sampling with probability no greater than 
$1/100$. 
In addition, Lemma~\ref{lemma-perturbOpt} and Lemma~\ref{lemma-optnet2} each 
fails to hold with probability no greater than 1/100. 
Finally, let $\widehat{r_2}$ be a random variable representing the number of 
rows actually chosen by our sampling schema in the second stage, and note that 
$\Ex{\widehat{r_2}} \le 2 r_2$. 
By Markov's inequality, it follows that $\widehat{r_2} > 40r_2$ with 
probability less than $1/20$. 
Thus, the second stage of our algorithm fails with probability less than 
$1/20+1/100+1/100+1/100 < 1/10$. 
By combining both stages, our algorithm fails to give a
$(1+\epsilon)$-approximation in the specified running time with a
probability bounded from above by $2/5+1/10 = 1/2$.

\vspace{-0.5cm}
\section{Extensions}
\label{sxn:extensions}
\vspace{-0.5cm}

In this section we outline several immediate extensions of our main
algorithmic result.

%
\para{Constrained \boldmath{$\ell_p$} regression}
Our sampling strategies are transparent to constraints placed on
$x$. In particular, suppose we constrain the output of our algorithm
to lie within a convex set $\mathcal{C} \subseteq \Real^m$. If there
is an algorithm to solve the constrained $\ell_p$ regression problem
$\min_{z\in\mathcal{C}}\norm{{\A}'x - b'}$, where ${\A}'\in
\Real^{s\times m}$ is of rank $d$ and $b' \in \Real^s$, in time
$\phi(s, m)$, then by modifying our main algorithm in a
straightforward manner, we can obtain an algorithm that gives a
$(1+\epsilon)$-approximation to the constrained $\ell_p$ regression
problem in time $O(nmd + nd^5\log n + \phi(40r_2,m))$.
%

\para{Generalized \boldmath{$\ell_p$} regression}
Our sampling strategies extend to the case of generalized $\ell_p$
regression: given as input a matrix ${\A} \in \Real^{n \times m}$ of
rank $d$, a target matrix $B \in \Real^{n\times p}$, and a real
number $p \in [1,\infty)$, find a matrix $X\in\Real^{m\times p}$
such that $\mnorm{\A X-B}_p$ is minimized. To do so, we generalize
our sampling strategies in a straightforward manner. The
probabilities $p_i$ for the first stage of sampling are the same as
before. Then, if $\hat{X}_c$ is the solution to the first-stage
sampled problem, we can define the $n \times p$ matrix $\hat{\rho} =
{\A}\hat{X}_c - B$, and define the second stage sampling
probabilities to be $q_i = \min\left(1,\max\{p_i, r_2
\|\row{\hat{\rho}}{i}\|^p_{p}/\mnorm{\hat{\rho}}^p_{p}\}\right)$.
Then, we can show that the $\hat{X}_{\OPT}$ computed from the
second-stage sampled problem satisfies $\mnorm{{\A}\hat{X}_{\OPT} -
B}_p
    \le (1+\epsilon)\min_{X\in\Real^{m\times p}}\mnorm{\A X - B}_p $,
with probability at least $1/2$.
%

\para{Weighted \boldmath{$\ell_p$} regression}
Our sampling strategies also generalize to the case of $\ell_p$ regression
involving weighted $p$-norms:
if $w_1, \ldots, w_m$ are a set of non-negative weights then the weighted
$p$-norm of a vector $x\in \Real^m$ may be defined as
$\norm{x}_{p,w} = \left(\sum_{i=1}^{m} w_i \abs{x_i}^p\right)^{1/p}$, and the
weighted analog of the matrix $p$-norm $\mnorm{\mathbf{\cdot}}_p$ may be
defined as
$\mnorm{{\U}}_{p,w}
   = \left(\sum_{j=1}^d \norm{\col{{\U}}{j}}_{p,w}\right)^{1/p}$.
Our sampling schema proceeds as before.
First, we compute a ``well-conditioned'' basis ${\U}$ for $\spn({\A})$ with
respect to this weighted $p$-norm.
The sampling probabilities $p_i$ for the first stage of the
algorithm are then $p_i = \min\left(1, r_1
w_i\norm{\row{{\U}}{i}}^p_p/\mnorm{{\U}}^p_{p,w}\right)$, and the
sampling probabilities $q_i$ for the second stage are $q_i =
\min\left(1,\max\{p_i, r_2
w_i\abs{\hat{\rho}_i}^p/\|\hat{\rho}\|_{p,w}^p\}\right)$, where
$\hat{\rho}$ is the residual from the first stage.
%

\para{General sampling probabilities}
More generally, consider any sampling probabilities of the form:
$p_i
   \ge \min
       \left\{1,
              \max
              \left\{\frac{\norm{\row{{\U}}{i}}_p^p}{\mnorm{{\U}}_p^p},
                     \frac{\abs{\left(\rho_{\OPT}\right)_i}^p}{{\cal Z}^p}
              \right\}r
       \right\}$,
where $\rho_{\OPT} = {\A}x_{\OPT} - b$ and $ r \ge \frac{36^p
d^k}{\epsilon^2}\left(d\ln(\frac{36}{\epsilon}) + \ln(200)\right)$
and where we adopt the convention that $\frac00 = 0$. Then, by an
analysis similar to that presented for our two stage algorithm, we
can show that, by picking $O(36^p d^{p+1}/\epsilon^{2})$ rows of
${\A}$ and the corresponding elements of $b$ (in a single stage of
sampling) according to these probabilities, the solution
$\hat{x}_{\OPT}$ to the sampled $\ell_p$ regression problem is a
$(1+\epsilon)$-approximation to the original problem, with
probability at least $1/2$. (Note that these sampling probabilities,
if an equality is used in this expression, depend on the entries of
the vector $\rho_{\OPT} = {\A}x_{\OPT}-b$; in particular, they
require the solution of the original problem. This is reminiscent of
the results of~\cite{DM06}. Our main two-stage algorithm shows that
by solving a problem in the first stage based on coarse
probabilities, we can refine our probabilities to approximate these
probabilities and thus obtain an $(1+\epsilon)$-approximation to the
$\ell_p$ regression problem more efficiently.)




\newpage


\newpage

\appendix

\section{Tail inequalities}

With respect to tail inequalities, we will use the following version of the
Bernstein's inequality.
\begin{theorem}[\cite{maurer, bernstein}]
\label{thm:tail_bounds}
Let $\{X_i\}_{i=1}^n$ be independent random variables with
$E[X_i^2] < \infty$ and $X_i \ge 0$.
Set $Y = \sum_i X_i$ and let $\gamma > 0$.
Then
\begin{align}
\Pr{Y \le E[Y] - \gamma} \ \le\
\exp\left(\frac{-\gamma^2}{2\sum_i E[X_i^2]}\right) \enspace .
\label{eqn:lowertail}
\end{align}
If $X_i - E[X_i] \le \Delta$ for all $i$, then with
$\sigma_i^2 = E[X_i^2]-E[X_i]^2$ we have
\begin{align}
\Pr{Y \ge E[Y] + \gamma} \ \le\ \exp\left(\frac{-\gamma^2}{2\sum_i\sigma_i^2 +
2\gamma\Delta/3}\right) \enspace .
\label{eqn:uppertail}
\end{align}
\end{theorem}

\section{Proofs for Section \ref{sxn:main_technical}}

\subsection{Proof of Theorem \ref{thm:preserve_subspace}}

\begin{proof}
For simplicity of presentation, in this proof we will generally drop the
subscript from our matrix and vector $p$-norms; i.e., unsubscripted norms will
be $p$-norms.
Note that it suffices to prove that, for all $x \in \Real^m$,
\begin{equation}
\label{eqn:preserve_subspace_prf}
(1- \eps)^p \norm{\A x}^p \leq \norm{S\A x}^p \leq (1+\eps)^p \norm{\A x}^p ,
\end{equation}
with probability $1-\delta$.
To this end, fix a vector $x \in \Real^m$, define the random variable
$X_i = \left(\Sa_{ii} |\row{\A}{i}  x |\right)^p$, and recall that
$\row{\A}{i} = \row{\U}{i} \tau $ since ${\A}={\U}\tau$.
Clearly, $\sum_{i=1}^{n} X_i = \norm{S\A x}^p$.
In addition, since $\Ex{X_i} = |\row{\A}{i}  x |^p$, it follows that
$\sum_{i=1}^{n} \Ex{X_i} = \norm{\A x}^p$.
To bound Equation~\eqref{eqn:preserve_subspace_prf}, first note that
\begin{equation}
\sum_{i=1}^{n} \left( X_i - \Ex{X_i} \right) \\
   = \sum_{i:p_i<1} \left( X_i - \Ex{X_i} \right)  .
\label{eqn:remove_p_eq_1}
\end{equation}
Equation~\ref{eqn:remove_p_eq_1} follows since, according to the definition
of $p_i$ in Equation~\eqref{eq-p-subspace}, $p_i$ may equal $1$ for some rows,
and since these rows are always included in the random sample, $X_i=\Ex{X_i}$
for these rows.
To bound the right hand side of Equation~\ref{eqn:remove_p_eq_1}, note that
for all $i$ such that $p_i<1$,
\begin{align}
\left|\row{{\A}}{i} x \right|^p/p_i
   & \leq \norm{\row{{\U}}{i}}_p^p\norm{\tau x}_q^p/p_i
     & \mbox{(by H\"{o}lders inequality)} \nonumber \\
  & \leq \mnorm{{\U}}_p^p \norm{\tau x}_q^p /r
    &\mbox{(by Equation~\eqref{eq-p-subspace})} \nonumber \\
\label{eqn:basic1}
  &\leq (\alpha\beta)^{p}\norm{\A x}^{p}/r
    &\mbox{(by Definition~\ref{def:good_basis} and
               Theorem~\ref{thm:good_basis})}\enspace.
\end{align}
From Equation~\eqref{eqn:basic1}, if follows that for each $i$ such that
$p_i<1$,
$$
X_i - \Ex{X_i}
   \le X_i
   \leq \abs{\row{{\A}}{i} x}^p/p_i
   \leq (\alpha\beta)^{p}\norm{\A x}^p/r ;
$$
Thus, we may define $\Delta = (\alpha\beta)^{p}\norm{\A x}^p/r$.
In addition, it also follows from Equation~\eqref{eqn:basic1} that
\begin{align*}
\sum_{i :p_i <1} \Ex{X_i^2}
   & = \sum_{i :p_i <1} \left|\row{\A}{i} x\right|^p
                        \frac{\left|\row{{\A}}{i} x \right|^p} {p_i}\\
   & \leq \frac{(\alpha\beta)^{p}\norm{\A x}^p}{r}
       \sum_{i :p_i <1} |\row{\A}{i} x |^p
     & \mbox{(by Equation~\eqref{eqn:basic1})}\\
   &\leq (\alpha\beta)^{p}\norm{\A x}^{2p}/r \enspace,
\end{align*}
from which it follows that
$\sum_{i:p_i<1} \sigma_i^2
   \leq \sum_{i:p_i<1} \Ex{X_i^2}
   \leq (\alpha\beta)^{p}\norm{\A x}^{2p}/r $.

To apply the upper tail bound in Theorem~\ref{thm:tail_bounds}, define
$\gamma = ((1+\eps/4)^p-1) \norm{\A x}^p$.
It follows that $\gamma^2 \geq (p\eps/4)^2\norm{\A x}^{2p}$ and also that
\begin{eqnarray*}
2\sum_{i:p_i<1} \sigma_i^2 + 2\gamma\Delta/3
   &\leq& 2(\alpha\beta)^{p}\norm{\A x}^{2p}/r+ 2((1+\eps/4)^p-1)(\alpha\beta)^{p}\norm{\A x}^{2p}/3r \\
   &\leq& 32^p(\alpha\beta)^{p}\norm{\A x}^{2p}/r ,
\end{eqnarray*}
where the second inequality follows by standard manipulations since
$\epsilon \le 1$ and since $p \ge 1$.
Thus, by Equation~\eqref{eqn:uppertail} of Theorem~\ref{thm:tail_bounds}, it
follows that
\begin{eqnarray*}
\Pr{ \norm{S\A x}^p > \norm{\A x}^p + \gamma }
   &=&    \Pr{\sum_{i:p_i<1} X_i > \Ex{\sum_{i:p_i<1} X_i} + \gamma}  \\
   &\leq& \exp\left(\frac{- \gamma^2}{2\sum_{i:p_i<1} \sigma_i^2 + 2\gamma\Delta/3 } \right)   \\
   &\leq& \exp\left( -\eps^{2}p^2 r / (\alpha\beta)^{p} 32^p\right)  .
\end{eqnarray*}
Similarly, to apply the lower tail bound of Equation~\eqref{eqn:lowertail} of
Theorem~\ref{thm:tail_bounds}, define
$\gamma = (1-(1-\eps/4)^p) \norm{\A x}^p$.
Since $\gamma \geq \eps\norm{\A x}^p/4$, we can follow a similar line of
reasoning to show that
\begin{eqnarray*}
\Pr{ \norm{S\A x}^p < \norm{\A x}^p - \gamma }
   &\leq& \exp\left( \frac{-\gamma^2}{2\sum_{i:p_i<1} \sigma_i^2 } \right)  \\
   &\leq& \exp\left(-\eps^{2} r /(\alpha\beta)^{p}32\right).
\end{eqnarray*}
Choosing
$r \ge 32^p(\alpha\beta)^p(d\ln(\frac{12}{\eps})+\ln(\frac{2}{\delta}))/(p^2\eps^{2})$,
we get that for every fixed $x$, the following is true with probability
at least $1 - \left(\frac{\eps}{12}\right)^d\delta$:
$$
(1- \eps/4)^p \norm{\A x}^p
   \leq \norm{S\A x}^p \leq (1+\eps/4)^p \norm{\A x}^p.
$$

Now, consider the ball $B = \{ y \in \Real^n : y = \A x, \norm{y} \le 1\}$
and consider an $\vareps$-net for $B$, with $\vareps=\eps/4$.
The number of points in the $\vareps$-net is $\left(\frac{12}{\eps}\right)^d$.
Thus, by the union bound, with probability $1 - \delta$,
Equation~\eqref{eqn:preserve_subspace_prf} holds for all points in the
$\vareps$-net.
Now, to show that with the same probability
Equation~\eqref{eqn:preserve_subspace_prf} holds for all points $y \in B$, let
$y^{*} \in B$ be such that $\abs{\norm{Sy} - \norm{y}}$ is maximized, and let
$\eta = \sup\{\abs{\norm{Sy} - \norm{y}} : y\in B\} $.
Also, let $y^{*}_{\vareps} \in B$ be the point in the $\vareps$-net that is
closest to $y^{*}$.
By the triangle inequality,
\begin{align*}
\eta = \abs{\norm{Sy^{*}} - \norm{y^{*}}}
& = \abs{\norm{Sy^{*}_\vareps + S(y^{*}-y^{*}_\vareps)}-\norm{y^{*}_\vareps + (y^{*}-y^{*}_\vareps)}} \\
& \le \abs{\norm{Sy^{*}_\vareps}+\norm{S(y^{*}-y^{*}_\vareps)}-
                  \norm{y^{*}_\vareps}+2\norm{y^{*}-y^{*}_\vareps}-\norm{y^{*}-y^{*}_\vareps}} \\
& \le \abs{\norm{Sy^{*}_\vareps}-\norm{y^{*}_\vareps}} + \abs{\norm{S(y^{*}-y^{*}_\vareps)}-\norm{y^{*}-y^{*}_\vareps}} + 2\norm{y^{*}-y^{*}_\vareps}\\
& \le \eps/4\norm{y^{*}_\vareps} + \eps\eta/4 + \eps/2 \enspace ,
\end{align*}
where the last inequality follows since
$\norm{y^{*} - y^{*}_{\vareps}} \leq \vareps$,
$(y^{*}-y^{*}_\vareps)/\vareps \in B$, and
\[\abs{\norm{S(y^{*}-y^{*}_\vareps)/\vareps} - \norm{(y^{*}-y^{*}_\vareps)/\vareps}} \le \eta \enspace .\]
Therefore, $\eta \le \epsilon$ since $\norm{y^{*}_\vareps} \le 1$ and since we
assume $\eps \le 1/7$.
Thus, Equation~\eqref{eqn:preserve_subspace_prf} holds for all points
$y \in B$, with probability at least $1-\delta$.
Similarly, it holds for any $y \in \Real^n$ such that $y = \A x$, since
$y/\norm{y} \in B$ and since
$\norm{S(y/\norm{y}) - y/\norm{y}} \le \epsilon$ implies that
$\norm{Sy - y} \le \epsilon \norm{y}$,
which completes the proof of the theorem.
\end{proof}

\section{Proofs for Section \ref{sxn:main_algorithm}}

As in the proof of
Theorem~\ref{thm:preserve_subspace}, unsubscripted norms will be
$p$-norms.

\subsection{Proof of Lemma \ref{lemma-perturb3Opt}}

\begin{proof}
Define $X_i = (\Sa_{ii} \abs{\row{{\A}}{i} x_{\OPT} - b_i} )^p$.
Thus, $\sum_i X_i= \norm{\Sa (\A x_{\OPT} - b)}^p $, and the first moment is
$\Ex{\sum_i X_i} = \norm{\A x_{\OPT} - b}^p = {\cal Z}$.
The lemma follows since, by Markov's inequality,
\begin{align*}
\Pr{\sum_i X_i >3^p\Ex{\sum_i X_i} } \leq \frac{1}{3^p},
\end{align*}
i.e., $\norm{\Sa (\A x_{\OPT} - b)}^p > 3^p \norm{\A x_{\OPT} - b}^p$,
with probability no more than $1/3^p$.
\end{proof}

\subsection{Proof of Lemma \ref{lemma-constantOpt}}

\begin{proof}
We will prove the contrapositive: If $\norm{{\A}\hat{x}_c - b} >
8\,{\cal Z}$, then $\norm{\Sa({\A}\hat{x}_c - b)} > 3\,{\cal Z}$. To
do so, note that, by Theorem \ref{thm:preserve_subspace}, and the
choice of $r_1$, we have that
\begin{align*}
\frac{7}{8} \norm{Ax}_p \leq \norm{SAx}_p \leq \frac{9}{8}
\norm{Ax}_p.
\end{align*}
Using this,
\begin{align*}
\norm{\Sa({\A}\hat{x}_{c} - b)}
   & \geq \norm{\Sa{\A}(\hat{x}_{c} - x_{\OPT})} - \norm{\Sa(\A x_{\OPT}-b)}
      & \mbox{(by the triangle inequality)}\\
   & \geq \frac{7}{8} \norm{{\A}\hat{x}_{c} - \A x_{\OPT}} - 3\,{\cal Z}
      & \mbox{(by Theorem~\ref{thm:preserve_subspace} and Lemma~\ref{lemma-perturb3Opt})}\\
   & \geq \frac{7}{8}  \left( \norm{{\A}\hat{x}_{c} - b } - \norm{\A x_{\OPT}- b} \right) - 3\,{\cal Z}
      & \mbox{(by the triangle inequality)}\\
   & > \frac{7}{8}  \left( 8\,{\cal Z} - {\cal Z} \right) - 3\,{\cal Z}
      & \mbox{(by the premise } \norm{{\A}\hat{x}_c - b} > 8\,{\cal Z}) \\
   &>3\,{\cal Z},
\end{align*}
which establishes the lemma.
\end{proof}
\subsection{Proof of Lemma \ref{lemma-perturbOpt}}

\begin{proof}
Define the random variable
$X_i = (\Sb_{ii} \abs{\row{{\A}}{i}x_{\OPT} - b_i} )^p$, and recall that
$\row{\A}{i} = \row{\U}{i} \tau $ since ${\A}={\U}\tau$.
Clearly, $\sum_{i=1}^n X_i = \norm{\Sb ({\A} x_{\OPT}-b)}^p$.
In addition, since $\Ex{X_i} = \abs{\row{{\A}}{i}x_{\OPT}-b_i}^p$, it follows
that $\sum_{i=1}^n \Ex{X_i} = \norm{ {\A} x_{\OPT}-b}^p$.
We will use Equation~\eqref{eqn:uppertail} of Theorem~\ref{thm:tail_bounds} to
provide a bound for
$\sum_i \left( X_i - \Ex{X_i} \right) =
\norm{\Sb ({\A} x_{\OPT}-b)}^p - \norm{ {\A} x_{\OPT}-b}^p$.

From the definition of $q_i$ in Equation~\eqref{eq-q}, it follows
that for some of the rows, $q_i$ may equal $1$ (just as in the proof of
Theorem~\ref{thm:preserve_subspace}).
Since $X_i=\Ex{X_i}$ for these rows,
$\sum_i \left( X_i - \Ex{X_i} \right)
   = \sum_{i:p_i <1} \left( X_i - \Ex{X_i} \right)$, and thus we will bound
this latter quantity with Equation~\eqref{eqn:uppertail}.
To do so, we must first provide a bound for $X_i-\Ex{X_i} \le X_i$ and for
$\sum_{i:p_i<1} \sigma_i^2 \le \sum_i \Ex{X_i^2}$.
To that end, note that:
\begin{align}
\nonumber
\abs{\row{\A}{i}(x_{\OPT}-\hat{x}_c)}
   & \leq \norm{\row{\U}{i}}_p \norm{\tau(x_{\OPT}-\hat{x}_c)}_q
    & \mbox{(by H\"{o}lders inequality)}\\
\nonumber
   & \leq \norm{\row{\U}{i}}_p \beta \norm{{\U}\tau(x_{\OPT}-\hat{x}_c)}_p
    & \mbox{(by Definition~\ref{def:good_basis}
             and Theorem~\ref{thm:good_basis})}\\
\nonumber
   & \leq \norm{\row{\U}{i}}_p \beta
          \left( \norm{{\A}x_{\OPT}-b} + \norm{{\A}\hat{x}_c-b} \right)
    & \mbox{(by the triangle inequality)}\\
\label{eqn1:lemma-perturb0pt}
   & \leq \norm{\row{\U}{i}}_p \beta 9 {\cal Z} \enspace ,
\end{align}
where the final inequality follows from the definition of ${\cal Z}$ and the
results from the first stage of sampling.
Next, note that from the conditions on the probabilities $q_i$ in
Equation~\eqref{eq-q}, as well as by Definition~\ref{def:good_basis} and the
output of the first-stage of sampling, it follows that
\begin{equation}
\label{eq-bothp}
   \frac{\abs{\hat{\rho}_i}^p}{q_i} \le \frac{\|\hat{\rho}\|^p}{r_2}
                                    \le \frac{8^p{\cal Z}^p}{r_2}
   \quad\text{and}\quad
   \frac{\norm{\row{{\U}}{i}}^p}{q_i} \le \frac{\mnorm{{\U}}^p}{r_2}
                                      \le \frac{\alpha^p}{r_2}
   \enspace,
\end{equation}
for all $i$ such that $q_i<1$.

Thus, since $X_i-\Ex{X_i} \leq X_i \leq \abs{\row{{\A}}{i}x_{\OPT}-b_i}^p/q_i$,
it follows that for all $i$ such that $q_i<1$,
\begin{align}
X_i -\Ex{X_i}
    & \leq \frac{2^{p-1}}{q_i}
           \left( \abs{\row{{\A}}{i} (x_{\OPT}-\hat{x}_c)}^p
                + \abs{\hat{\rho}_i}^p
           \right)
    & \mbox{(since } \hat{\rho} = {\A}\hat{x}_c - b \mbox{ )} \\
    & \leq 2^{p-1}
           \left( \frac{ \norm{\row{\U}{i}}^p_p \beta^p 9^p {\cal Z}^p }{q_i}
               +  \frac{ \abs{\hat{\rho}_i}^p }{q_i}
           \right)
      & \mbox{(by Equation~\eqref{eqn1:lemma-perturb0pt})} \nonumber \\
    & \leq 2^{p-1}
           \left( \alpha^p \beta^p 9^p {\cal Z}^p +  8^p {\cal Z}^p
           \right)/r_2
      & \mbox{(by Equation~\eqref{eq-bothp})} \nonumber \\
\label{eqn2:lemma-perturb0pt}
    & \leq c_p (\alpha\beta)^p {\cal Z}^p /r_2 \enspace ,
\end{align}
where we set $c_p = 2^{p-1}(9^p + 8^p)$.
Thus, we may define $\Delta = c_p(\alpha\beta)^p{\cal Z}^p/r_2$. In
addition, it follows that
\begin{align}
\nonumber
\sum_{i:q_i<1} \Ex{X_i^2}
   & =\ \sum_{i:q_i<1}\abs{\row{{\A}}{i}x_{\OPT} -b_i}^p
                      \frac{\abs{\row{{\A}}{i}x_{\OPT} -b_i}^p}{q_i} \\
  & \le\ \Delta
          \sum_i\abs{\row{{\A}}{i} x_{\OPT} - b_i}^p \nonumber
    & \mbox{(by Equation~\eqref{eqn2:lemma-perturb0pt})} \\
  & \le\ c_p(\alpha\beta)^p{\cal Z}^{2p}/r_2 \label{eqn3:lemma-perturb0pt} \enspace .
\end{align}
To apply the upper tail bound of Equation~\eqref{eqn:uppertail} of
Theorem~\ref{thm:tail_bounds}, define
$\gamma = ((1+\epsilon)^p-1) {\cal Z}^p$.
We have $\gamma \ge p\epsilon\,{\cal Z}^{p}$, 
and since $\epsilon \le 1/7$, we also have $\gamma \le \left(\left(\frac87\right)^p - 1\right){\cal Z}^p$.
Hence, by Equation~\eqref{eqn:uppertail} of Theorem~\ref{thm:tail_bounds}, it
follows that
\begin{align*}
\ln\Pr{\norm{\Sb(\A x_{\OPT} - b)}^p > \norm{\A x_{\OPT} - b}^p + \gamma}
   & \leq \frac{-\gamma^2}
               {2\sum_{i:p_i<1} \sigma_i^2 + 2\gamma\Delta/3 } \\
   & \leq \frac{-p^2\epsilon^{2} r_2}{36^p(\alpha\beta)^p} \enspace .
\end{align*}
Thus, $\Pr{\norm{\Sb(\A x_{\OPT} - b)} > (1+\eps){\cal Z}} \le \exp
\left( \frac{-p^2\epsilon^{2} r_2}{36^p(\alpha\beta)^p}\right)$, from which
the lemma follows by our choice of $r_2$.
\end{proof}

\subsection{Proof of Lemma \ref{lemma-3opt}}

\begin{proof}
By two applications of the triangle inequality, it follows that
\begin{align*}
\norm{{\A}\hat{x}_{\OPT}-{\A}\hat{x}_c}
   & \leq \norm{\A \hat{x}_{\OPT} - {\A}x_{\OPT}}
        + \norm{{\A}x_{\OPT} - b}
        + \norm{{\A}\hat{x}_c- b}\\
   & \leq \norm{\A \hat{x}_{\OPT} - {\A}x_{\OPT}} + 9{\cal Z} \enspace ,
\end{align*}
where the second inequality follows since
$\norm{{\A}\hat{x}_c-b} \le 8\,{\cal Z}$ from the first stage of sampling and
since ${\cal Z}=\norm{{\A}x_{\OPT} - b}$.
In addition, we have that
\begin{align*}
\norm{\A x_{\OPT} - {\A}\hat{x}_{\OPT}}
   & \le\ \frac{1}{(1-\epsilon)}\norm{\Sb({\A}\hat{x}_{\OPT} - \A x_{\OPT})}
      & \mbox{(by Theorem~\ref{thm:preserve_subspace})} \\
   & \le\ (1+\epsilon)\left(\norm{\Sb({\A}\hat{x}_{\OPT} - b)} + \norm{\Sb(\A x_{\OPT} - b)}\right)
      & \mbox{(by the triangle inequality)}\\
   & \le\ 2(1+\epsilon)\norm{\Sb(\A x_{\OPT} - b)} \\
   & \le\ 2(1+\epsilon)^2\norm{\A x_{\OPT} - b}
      & \mbox{(by Lemma~\ref{lemma-perturbOpt})} \enspace ,
\end{align*}
where the third inequality follows since $\hat{x}_{\OPT}$ is optimal for the
sampled problem.
The lemma follows since $\epsilon \le 1/7$.
\end{proof}

\subsection{Proof of Lemma \ref{lemma-optnet2}}

\begin{proof}
Fix a given point $y^{*}_\varepsilon = \A x^{*}_\varepsilon \in B_\varepsilon$.
We will prove the contrapositive for this point, i.e., we will prove that if
$\norm{\A x^{*}_\varepsilon - b} > (1+6\epsilon){\cal Z}$, then
$\norm{\Sb(\A x^{*}_\varepsilon - b)} > (1+3\epsilon){\cal Z}$, with
probability at least $1-\frac{1}{100}\left(\frac{\epsilon}{36}\right)^d$.
The lemma will then follow from the union bound.

To this end, define the random variable
$X_i = (\Sb_{ii}\abs{\row{{\A}}{i} x^{*}_\varepsilon - b_i})^p$, and recall
that $\row{{\A}}{i} = \row{{\U}}{i}\tau $ since ${\A} = {\U}\tau$.
Clearly,
$\sum_{i=1}^n X_i = \norm{\Sb ({\A} x^{*}_\varepsilon-b)}^p$.
In addition, since $\Ex{X_i} = \abs{\row{{\A}}{i}x^{*}_\varepsilon-b_i}^p$, it
follows that $\sum_{i=1}^n \Ex{X_i} = \norm{ {\A} x^{*}_\varepsilon-b}^p$.
We will use Equation~\eqref{eqn:lowertail} of Theorem~\ref{thm:tail_bounds} to
provide an upper bound for the event that
$\norm{\Sb({\A}x^{*}_\varepsilon-b)}^p \le \norm{{\A}x^{*}_\varepsilon-b}^p-\gamma$,
where $\gamma = \norm{\A x^{*}_\varepsilon - b}^p - (1+3\epsilon)^p{\cal Z}^p$,
under the assumption that
$\norm{\A x^{*}_\varepsilon - b} > (1+6\epsilon){\cal Z}$.

From the definition of $q_i$ in Equation~\eqref{eq-q}, it follows
that for some of the rows, $q_i$ may equal $1$ (just as in the proof of
Theorem~\ref{thm:preserve_subspace}).
Since $X_i=\Ex{X_i}$ for these rows,
$\sum_i \left( X_i - \Ex{X_i} \right)
   = \sum_{i:p_i <1} \left( X_i - \Ex{X_i} \right)$, and thus we will bound
this latter quantity with Equation~\eqref{eqn:lowertail}.
To do so, we must first provide a bound for $\sum_{i:p_i <1} \Ex{X_i^2}$.
To that end, note that:
\begin{align}
\nonumber
\abs{\row{\A}{i}(x^{*}_\varepsilon-\hat{x}_c)}
   & \leq \norm{\row{\U}{i}}_p \norm{\tau(x^{*}_\varepsilon-\hat{x}_c)}_q
    & \mbox{(by H\"{o}lders inequality)}\\
\nonumber
   & \leq \norm{\row{\U}{i}}_p \beta \norm{{\U}\tau(x^{*}_\varepsilon-\hat{x}_c)}_p
    & \mbox{(by Definition~\ref{def:good_basis}
             and Theorem~\ref{thm:good_basis})}\\
\label{eqn1:lemma-optnet2}
   & \leq \norm{\row{\U}{i}} \beta 12 {\cal Z} \enspace ,
\end{align}
where the final inequality follows from the radius of the high-dimensional
ball in which the $\varepsilon$-net resides.
From this, we can show that
\begin{align}
\nonumber
\frac{ \abs{\row{\A}{i}x^{*}_\varepsilon-b_i} }{ q_i }
   & \leq \frac{2^{p-1}}{q_i}
          \left( \abs{\row{\A}{i}x^{*}_\varepsilon-\row{\A}{i}\hat{x}_c}^p +
                 \abs{\hat{\rho}_i}^p
          \right)
    & \mbox{(since } \hat{\rho} = {\A}\hat{x}_c - b \mbox{ )} \\
\nonumber
   & \leq 2^{p-1}
          \left( \frac{ \norm{\row{\U}{i}}^p 12^p \beta^p {\cal Z}^p }{q_i} +
                 \frac{\abs{\hat{\rho}_i}^p}{q_i}
          \right)
    & \mbox{(by Equation~\eqref{eqn1:lemma-optnet2})} \\
\nonumber
   & \leq 2^{p-1}
          \left( \alpha^p 12^p \beta^p {\cal Z}^p + 8^p {\cal Z}^p
          \right)/r_2
    & \mbox{(by Equation~\eqref{eq-bothp})} \\
\label{eqn3:lemma-optnet2}
   & \leq 24^p (\alpha\beta)^p {\cal Z}^p /r_2 \enspace .
\end{align}
Therefore, we have that
\begin{align}
\nonumber
\sum_{i:q_i<1} \Ex{X_i^2}
   & =\ \sum_{i:q_i<1}\abs{\row{{\A}}{i}x^{*}_\varepsilon -b_i}^p
                      \frac{\abs{\row{{\A}}{i}x^{*}_\varepsilon -b_i}^p}{q_i} \\
   & \le\ \frac{24^p(\alpha\beta)^p{\cal Z}^p}{r_2}
          \sum_i\abs{\row{{\A}}{i} x^{*}_\varepsilon - b_i}^p \nonumber
    & \mbox{(by Equation~\eqref{eqn3:lemma-optnet2})} \\
   & \le\ 24^p(\alpha\beta)^p
          \norm{\A x^{*}_\varepsilon - b}^{2p} /r_2 \label{eq-var3}.
\end{align}
To apply the lower tail bound of Equation~\eqref{eqn:lowertail} of
Theorem~\ref{thm:tail_bounds}, define
$\gamma = \norm{\A x^{*}_\varepsilon - b}^p - (1+3\epsilon)^p{\cal Z}^p$.
Thus, by Equation~\eqref{eq-var3} and by Equation~\eqref{eqn:lowertail} of
Theorem~\ref{thm:tail_bounds} it follows that
\begin{align*}
\ln\Pr{\norm{\Sb(\A x^{*}_\varepsilon - b)}^p \le
(1+3\epsilon)^p{\cal Z}^p} & \le\ \frac{-r_2(\norm{\A
x^{*}_\varepsilon - b}^p - (1+3\epsilon)^p{\cal Z}^p)^2}
                 {24^p(\alpha\beta)^p\norm{\A x^{*}_\varepsilon - b}^{2p}}\\
& \le\ \frac{-r_2}{24^p(\alpha\beta)^p}
\left(1-\frac{(1+3\epsilon)^p{\cal Z}^p}{\norm{\A x^{*}_\varepsilon - b}^p} \right)^2 \\
& <\ \frac{-r_2}{24^p(\alpha\beta)^p}
\left(1-\frac{(1+3\epsilon)^p{\cal Z}^p}{(1+6\epsilon)^p{\cal
Z}^p}\right)^2
  & \mbox{(by the premise)} \\
   & \le\ \frac{-r_2\epsilon^2}{24^p(\alpha\beta)^p}
  & \mbox{(since $\epsilon \le 1/3$).} 
\end{align*}
Since $r_2 \ge
24^p(\alpha\beta)^{p}(d\ln(\frac{36}{\epsilon}) +
\ln(200))/\epsilon^2$, it follows that $\norm{\Sb(\A
x^{*}_\varepsilon - b)} \le (1+3\epsilon){\cal Z}$, with probability
no greater than $\frac{1}{200} \left( \frac{\epsilon}{36}
\right)^d$. Since there are no more than $ \left(\frac{36}{\epsilon}
\right)^d$ such points in the $\varepsilon$-net, the lemma follows
by the union bound.
\end{proof}

\subsection{Proof of Lemma \ref{lemma-nn}}

\begin{proof}
We will prove the contrapositive:
If $\norm{{\A}\hat{x}_{\OPT} - b} > (1+7\epsilon){\cal Z}$
then $\norm{\Sb({\A}\hat{x}_{\OPT} - b)} > (1+\epsilon){\cal Z}$.
Since ${\A}\hat{x}_{\OPT}$ lies in the ball $B$ defined by
Equation~\eqref{eq-affineball} and since the $\varepsilon$-net is constructed
in this ball, there exists a point $y_\varepsilon = \A x_\varepsilon$,
call it ${\A}x^{*}_\varepsilon$, such that
$\norm{{\A}\hat{x}_{\OPT}-{\A}x^{*}_\varepsilon} \le \epsilon\,{\cal Z}$.
Thus,
\begin{align*}
\norm{{\A}x^{*}_\varepsilon - b}
   & \ge\ \norm{{\A}\hat{x}_{\OPT} - b} - \norm{{\A}x^{*}_\varepsilon - {\A}\hat{x}_{\OPT}}
      & \mbox{(by the triangle inequality)} \\
   & \ge\ (1+7\epsilon){\cal Z} - \epsilon{\cal Z}
      & \mbox{(by assumption and the definition of } {\A}x^{*}_\varepsilon \mbox{ )}\\
   & =\ (1+6\epsilon){\cal Z} \enspace.
\end{align*}
Next, since Lemma~\ref{lemma-optnet2} holds for all points ${\A}x_\varepsilon$
in the $\varepsilon$-net, it follows that
\begin{equation}
\label{eqn1:lemma-nn}
\norm{\Sb({\A}x^{*}_\varepsilon - b)} \ >\ (1+3\epsilon){\cal Z} \enspace.
\end{equation}
Finally, note that
\begin{align*}
\norm{\Sb({\A}\hat{x}_{\OPT} - b)}
   & \ge\ \norm{\Sb({\A}x^{*}_\varepsilon - b)} - \norm{\Sb {\A}(x^{*}_\varepsilon - \hat{x}_{\OPT})}
      & \mbox{(by the triangle inequality)} \\
   & >\ (1+3\epsilon){\cal Z} - (1+\epsilon)\norm{{\A}(x^{*}_\varepsilon - \hat{x}_{\OPT})}
      & \mbox{(by Equation~\eqref{eqn1:lemma-nn} and Theorem~\ref{thm:preserve_subspace})} \\
   & >\ (1+3\epsilon){\cal Z} - (1+\epsilon)\epsilon\,{\cal Z}
      & \mbox{(by the definition of ${\A}\hat{x}_\varepsilon$)} \\
   & >\ (1+\epsilon){\cal Z} \enspace,
\end{align*}
which establishes the lemma.
\end{proof}

\end{document}